\documentclass[runningheads]{llncs}
\usepackage[T1]{fontenc}
\usepackage{graphicx}
\usepackage{hyperref}
\usepackage{color}

\usepackage{xspace}

\usepackage{amssymb,amsmath}
\usepackage{enumitem}
\usepackage[boxed, linesnumbered]{algorithm2e}
\usepackage{enumitem}
\usepackage{tikz}
\usetikzlibrary{automata}
\usetikzlibrary{arrows.meta, arrows, automata, calc, er, matrix, mindmap, positioning, shapes.geometric, shadows}
\begin{document}

\newcommand{\Reals}{\mathbb{R}}
\newcommand{\Nats}{\mathbb{N}}
\newcommand{\Ints}{\mathbb{Z}}
\newcommand{\Rats}{\mathbb{Q}}

\newcommand{\nba}{\textrm{NBA}}
\newcommand{\FGWNBA}{\textrm{WNBA}\ensuremath{_{\textsf{FG}}}\xspace}

\newcommand{\Aut}[1]{\ensuremath{\mathcal{#1}}}
\newcommand{\size}[1]{\ensuremath{|#1|}}

\newcommand{\BigO}[1]{\ensuremath{\mathcal{O}(#1)}}

\newcommand{\bs}[1]{\ensuremath{\overline{#1}}}
\newcommand{\relu}{\ensuremath{\mathit{relu}}}

\newcommand{\words}{\mathcal{W}_\Sigma}
\newcommand{\slice}{\mathbin{:}} 

\newcommand{\enc}{\ensuremath{\mathit{enc}}}
\newcommand{\dec}{\ensuremath{\mathit{dec}}}
\newcommand{\code}{\enc}

\newcommand{\vks}[1]{\ensuremath{|\!|#1|\!|}}

\newcommand{\twos}[2]{\ensuremath{\begin{bmatrix}#1\\#2\end{bmatrix}}}
\newcommand{\threes}[3]{\ensuremath{\begin{bmatrix}#1\\#2\\#3\end{bmatrix}}}

\newcommand{\WF}[2]{\ensuremath{\mathit{WF}_{#1}^{#2}}}

\newcommand{\proj}[2]{\ensuremath{{#1}{\downarrow}_{#2}}}

\newcommand{\join}{\mathrm{join}}
\newcommand{\project}{\mathrm{project}}
\newcommand{\interleave}{\mathrm{interleave}}

\newcommand{\measure}[1]{\ensuremath{|\!| #1 |\!|}}

\newcommand{\build}{\mathrm{build}}

\newcommand{\tofsa}{\mathrm{to\FGWNBA}}

\newcommand{\pwl}{\mathrm{PWL}}

\newcommand{\TODOcomment}[2]{%
        \stepcounter{TODOcounter#1}%
        {\scriptsize\bf$^{(\arabic{TODOcounter#1})}$}%
        \marginpar[\fbox{
                \parbox{2cm}{\raggedleft
                        \scriptsize$^{({\bf{\arabic{TODOcounter#1}{#1}}})}$%
                        \scriptsize #2}}]%
        {\fbox{\parbox{2cm}{\raggedright 
                                \scriptsize$^{({\bf{\arabic{TODOcounter#1}{#1}}})}$%
                                \scriptsize #2}}}
}%

\newcommand{\simpleTODOcomment}[2]{%
        \stepcounter{TODOcounter#1}%
        {\bf
                \scriptsize({\arabic{TODOcounter#1}~{#1}})
                {\bfseries{TODO:} #2}
        }
}

\newcounter{TODOcounter}
\newcommand{\TODO}[1]{\TODOcomment{}{#1}}
\newcommand{\TODOX}[1]{\simpleTODOcomment{}{#1}}

\title{Verifying And Interpreting Neural Networks using Finite Automata}
%
%
\author{Marco S\"alzer\orcidID{0000-0002-8012-5465} \and
Eric Alsmann\and
Florian Bruse\orcidID{0000-0001-6800-7135}\and
Martin Lange\orcidID{0000-0002-1621-0972}}

\authorrunning{M. S\"alzer et al.}
%
\institute{School of Electr.\ Eng.\ and Computer Science, University of Kassel, Germany
\email{\{marco.saelzer,eric.alsmann,florian.bruse,martin.lange\}@uni-kassel.de}}
\maketitle              
\begin{abstract}
Verifying properties and interpreting the behaviour of deep neural networks (DNN) is an important
task given their ubiquitous use in applications, including safety-critical ones, and their
black-box nature.
We propose an automata-theoric approach to tackling problems arising in DNN analysis. We show
that the input-output behaviour of a DNN can be captured precisely by a (special) weak B\"uchi 
automaton and we show how these can be used to address common verification 
and interpretation tasks of DNN like adversarial robustness or minimum sufficient reasons.
\end{abstract}


\section{Introduction}

Deep Neural Networks (DNN), trained using task-oriented and precisely crafted techniques, are the driving force of all modern deep learning applications, which have produced astonishing results: highly-developed driving assistants \cite{GrigorescuTCM20_selfdriving}, 
the overcoming of language barriers due to neural machine translation \cite{OtterMK21},
far-reaching support in early disease detection \cite{LitjensKBSCGLGS17}, 
the creation of inspiring art from textual user inputs \cite{RameshPGGVRCS21_dalle,RDNCC22_dalle2}, etc. 

Their striking performance comes with a downside: they are a black box.
While it is easy to describe structure and parameters of a DNN, it is hard to obtain reliable
predictions for or explanation of their behaviour. 
Deep learning techniques need to be reliable, though, especially in safety-critical applications. However, certifying that some DNN satisfies 
specific safety properties, formally called \emph{verifying} these properties, is difficult. The verification of a 
common safety property for a DNN $N$ is informally described by the question 
``is there an input $\bs{x}$ of interest such that the output $N(\bs{x})$ has some unwanted characteristics?'' 
The corresponding decision problem, formally called \emph{output reachability}, is NP-complete \cite{KBDJK17_reluplex}, 
even for completely shallow DNN and simple specifications of relevant inputs and outputs \cite{SalzerL21_npcompl}.
Furthermore, DNN-based applications require comprehensible explanations of the outputs generated by a DNN
due to legal, safety and ethical concerns. There is a need for techniques 
giving understandable explanations for DNN behaviour; this is formally known as \emph{interpreting} DNN. A typical 
interpretation task for some DNN $N$ and an input-output pair $(\bs{x}, N(\bs{x}))$ is to answer the question
``which features of $\bs{x}$ are the relevant ones leading to the output $N(\bs{x})$?''  A corresponding
decision problem, called the \textsc{MinimumSufficientReason} problem, is known to be $\Sigma^P_2$-complete \cite{BarceloM0S20_interprcompl}.

We propose an approach based on finite-state-automata for tackling challenges arising from the black-box nature of DNN.
A DNN $N$ computes a function of type $\Reals^m \rightarrow \Reals^n$ for some $m,n \in \Nats$, 
which induces a relation $R_N \subseteq \Reals^m \times \Reals^n$. Using an appropriate encoding, 
$R_N$ can be represented by a set of infinite words over an alphabet of \emph{$(m+n)$-track symbols} 
of the form $(a_1, \dotsc, a_m, b_1, \dotsc, b_n)$ where $a_i,b_i$ are taken from an alphabet like $\{0,1,.\,,+,-\}$.
A finite-state automaton $\Aut{A}$ over such \emph{$(m+n)$-track words} 
can be seen as a (synchronised) transducer between input symbols $(a_1,\ldots,a_m)$ and 
output symbols $(b_1,\ldots,b_n)$. Synchronicity guarantees regularity of the auomata's languages \cite{Berstel79}.

We present a complete construction of a weak B\"uchi automaton of exponential size that recognises the input-output behaviour 
of a given DNN. Weak B\"uchi automata are known to allow for more efficient algorithms than general B\"uchi automata, as they can
also be seen as co-B\"uchi automata. In fact, we show that not even the full power of weak automata is needed but a special
subclass suffices. We then show how most relevant problems regarding the verification and interpretation of DNN can be addressed 
using this construction and automata-theoretic machinery. It turns out that the exponential blowup in the translation is unavoidable,
(unless P = NP) as it can be used to decide output reachability.

In Sect.~\ref{sec:prelim} we give preliminary definitions regarding DNN, encodings of reals and Büchi automata. 
The core contribution, the construction of a special kind of weak Büchi automaton capturing the behaviour of DNN, is done in 
Sect.~\ref{sec:translation}.
In Sect.~\ref{sec:verification}, we introduce common verification and interpretation
problems regarding DNN and show that they can be tackled using the translation from DNN to automata.
In Sect.~\ref{sec:outlook}, we summarise and discuss possible future work.
Proof details for the technical results are deferred to App.~\ref{sec:proofs}.

\paragraph*{Related work.}
The work presented here falls into the intersection of neural-network-based machine learning on the one hand, and automata
theory on the other. Extra focus is on the use of automata-theoretic tools for tackling challenges on the machine learning side.
Most ongoing research in this area is focused on the combination of automata and so-called recurrent neural networks (RNN),
a model for processing sequential data \cite{WeissGY18_rnnauto,KhmelnitskyNRXB21_rnnauto,AyacheEG18_rnnauto,MayrYV21_rnnauto}.
Additionally, there was extensive research on automata and RNN in earlier days of neural network analysis. A good 
overview of this is given in \cite{Jacobsson_rnnautosurvey}. 
The common underlying theme there is to obtain finite-state automata, 
often DFA, which capture the dynamic behaviour of RNN.
The goal of our approach here is similar, yet there are two fundamental differences: 
first, the techniques mentioned above only work for RNN, while our approach can be 
applied to more general neural network models, including linear layers with piece-wise linear activations. 
It is open how far the approach generalises.
Second, the automata derived from RNN work on sequences of data points, where each single data point is a symbol. 
Finite alphabets are obtained by finitely partitioning the real-valued input space of an RNN.
Our approach yields automata working on single, encoded data points. By using 
nondeterministic Büchi automata (NBA), we retain full precision regarding the input space.
Xu et al.\ \cite{XuWQH21__nntoauto} present an active-learning based algorithm for extracting DFA
from neural network classifiers. Similar to our approach, these DFA work on encoded inputs of the neural
network. Since they use abstraction techniques, the resulting on-tape automata only approximate the behaviour of the 
neural network.

Use cases of our translation from DNN to finite-state automata explored in this paper include verification 
and interpretation of DNN. A comprehensive survey on DNN verification is given
by Huang et al.\ \cite{HKRSSTWY20_survey}, one on the state-of-the-art regarding DNN interpretation
is given by Zhang et al.\ \cite{ZhangTLT21_surveyinterp}.
  
It is also not hard to see that the problems in DNN verification and interpretation considered here can be expressed in the 
(decidable) theory of the reals with addition and multiplication by rational constants. Interestingly, 
weak B\"uchi automata -- which avoid most intrinsically difficult constructions for general
B\"uchi automata -- can be used to decide this theory \cite{BoigelotRW98_realnba,BoigelotJW05}. We remark, though, that DNN do not need 
the full power of this logic but only the existential-positive fragment. It is therefore reasonable to construct 
weak B\"uchi automata for DNN directly instead of going through the more powerful general theory of the reals.

\section{Preliminaries}
\label{sec:prelim}

For a $k$-dimensional vector $v \in A^k$ with $k \ge 1$ and some set $A$, we denote its components by $v_1,\ldots,v_k$ respectively.
Sometimes, we write vectors like $\bs{x},\bs{v},\ldots$ to stress their vector nature.


\subsubsection*{(Deep) Neural Networks.} A \emph{(DNN-)node} is a function $v\colon\Reals^k\to\Reals$ with $v(\bs{x}) 
= \sigma(\sum_{i=1}^{k} c_i x_i + b)$, where $k$ is the \emph{input dimension}, the $c_i \in \Rats$ are called \emph{weights}, 
$b \in \Rats$ is the \emph{bias} and $\sigma: \Reals \to \Reals$ is the \emph{activation function} of $v$.\footnote{The
literature allows weights and biases from $\Reals$. Since we study effective translations, DNN need to be finitely 
representable so we require the values to be rational.} A common activation function is the piecewise 
linear \emph{ReLU} function (for \emph{Rectified Linear Unit}), defined as $\relu(x) = \max(0,x)$.
A \emph{(DNN-)layer} $l$ is a tuple of some $n$ nodes $(v_1, \dotsc, v_n)$ where each node has the same input dimension $m$ and the same activation function. It computes the function $l\colon\Reals^m\to\Reals^n$ via $l(\bs{x}) = (v_1(\bs{x}), \dotsc, v_n(\bs{x}))$. We call 
$m$ the \emph{input} and $n$ the \emph{output dimension} of $l$.
A \emph{Deep Neural Network (DNN)} $N$ consists of $k$ layers $l_1, \dotsc, l_k$, where $l_1$ has input dimension $m$, the output dimension of $l_i$ is equal to the input 
dimension of $l_{i+1}$ for $i < k$ and the output dimension of $l_k$ is $n$. The DNN $N$ computes a function from $\Reals^m$ to $\Reals^n$ by $N(\bs{x}) = l_k(l_{k-1}( \ldots l_1(\bs{x}) \ldots ))$.

In order to estimate the asymptotic complexity of the proposed translation from DNN to finite-state automata, we introduce the 
following (approximate) size measures. For $c \in \Rats$ let $\measure{c} := \log |n| + \log d$ where $d$ is the smallest positive 
natural number s.t.\ $\frac{n}{d} = c$ with $n \in \Ints$. Accordingly, we define the size of a DNN-node $v$ computing
$\sum_{i=1}^{k} c_i x_i +b$ as $\measure{v} = \sum_{i=1}^k \measure{c_i} + \measure{b}$ and the size of a DNN $N$ 
with a total of $k$ nodes $v_1, \dotsc, v_k$ as $\measure{N} = \sum_{i=1}^{k} \measure{v_i}$.



\subsubsection*{Weak B\"uchi Automata.}
Let $\Sigma$ be an alphabet. As usual, let $\Sigma^*$ and $\Sigma^\omega$ denote the set of all finite, resp.\ infinite words over
$\Sigma$. 
A \emph{nondeterministic B\"uchi automaton} (NBA) is a tuple $\mathcal{A} = (Q,\Sigma,q_0,\delta,F)$ s.t.\ $Q$ is 
a finite set of states, $\Sigma$ is the underlying alphabet, $q_0 \in Q$ is a designated starting state, 
$\delta \subseteq Q \times \Sigma \times Q$ is the transition relation, and $F \subseteq Q$ is a designated set of accepting
states. The \emph{size} of $\Aut{A}$ is measured as $\size{\Aut{A}} := \size{Q}$. 
A run on an infinite word 
$w = a_0a_1\ldots$ is an infinite sequence $\rho = q_0,q_1,\ldots$ starting in the initial state and satisfying 
$(q_i,a_i,q_{i+1}) \in \delta$ for all $i \ge 0$. It is accepting if $q_i \in F$ for infinitely many $i$. The language of an 
NBA $\Aut{A}$ is $L(\Aut{A}) = \{ w \in \Sigma^\omega \mid \text{there is an accepting run of } \Aut{A} \textit{ on } w\}$. 

A \emph{weak (nondeterministic) B\"uchi automaton} (WNBA) is an NBA whose state set $Q$ can be partitioned into strongly connected
components (SCC) such that for each SCC $S \subseteq Q$ we have $S \subseteq F$ or $S \cap F = \emptyset$, i.e.\ each SCC
either consists of accepting states or non-accepting states only. It is known that WNBA are less expressive than NBA, for example 
there is no WNBA accepting $(a^*b)^\omega$. For the purposes developed here, namely the recognition of relations
of real numbers defined by arithmetical operations, weak NBA suffice, which has been observed before \cite{BoigelotJW05}.
The benefit of using WNBA comes from better algorithmic properties: whilst, for example, determinisation is notoriously
difficult for general NBA, it is much simpler for WNBA as they can also be seen as co-B\"uchi automata that are easier
to determinise \cite{Miyano:1984:AFA}. Likewise, minimisation is
quite important for practical applications, and just like determinisation, minimisation of general B\"uchi automata is more
difficult than it is for automata on finite words, while algorithms for those can typically be lifted to weak B\"uchi automata,
cf.\ \cite{Loding:2001:EMD}.

In fact, it turns out that we do not even need the full power of WNBA either. 
An \emph{eventually-always weak nondeterministic B\"uchi automaton} (\FGWNBA) is a WNBA such that every path through its 
state set contains at most one transition 
from a non-accepting to an accepting state and no transitions from accepting to non-accepting ones. In other words, every
accepting run is of the form $(Q \setminus F)^* F^\omega$. Furthermore, \FGWNBA are closed under unions and intersections, using
the usual product construction and appropriate sets of accepting states. 
Later we reduce decision problems on DNN to automata-theoretic ones. We therefore need to argue that the corresponding problems
on the automata side are (efficiently) decidable. For the DNN problems considered here, language emptiness suffices, and more
complex problems like inclusion are not needed. The following is well-known for (weak) B\"uchi automata.
 
\begin{proposition}
\label{prop:emptiness}
Emptiness of a \FGWNBA $\Aut{A}$ is decidable in time linear in $\size{\Aut{A}}$.
\end{proposition}

\subsubsection*{Encodings of reals.}
In the following, let $\Sigma = \{+,-,.\,,0,1\}$ unless stated explicitly otherwise. 
A word $w = s a_{n-1} \ldots a_0 . b_0 b_1 \ldots$ with $n \ge 0$, $s \in \{+,-\}$, $a_i,b_i \in \{0,1\}$ uniquely encodes a real value
$\dec(w) := (-1)^{\mathrm{sign}(s)} \cdot (\sum_{i=0}^{n-1} a_i \cdot 2^i + \sum_{i=0}^\infty b_i \cdot 2^{-(i+1)})$
where $\mathrm{sign}(s) = 0$ if $s = +$ and $\mathrm{sign}(s) = 1$ otherwise, and $\sum_{i=0}^{-1} \varphi_i = 0$. Note that 
the infinite sum on the right is always converging. 
Moreover, while the decoding $\dec(w)$ of a word $w$ is unique, the encoding $\enc(r)$ of any $r \in \Reals$ as such a word 
in binary representation is not necessarily unique, for three reasons: leading zeros change the word representation but not the underlying value,
both $+0.0^\omega$ and $-0.0^\omega$ represent the same value, namely $0$, and any number whose representation has a suffix of the form 
$10^\omega$ (possibly including a dot) also can be written with the suffix $01^\omega$ instead. For instance, the number $12$ has representations 
$+1100.0^\omega$ and $+1011.1^\omega$. 

\subsubsection*{\FGWNBA for relations of reals.}
Let $k \geq 1$. We denote with $\Sigma^k$ the alphabet consisting of all $k$-vectors of letters from $\Sigma$, using both 
vertical (as below) and horizontal vector notation (like $[s_1,\ldots,s_k]$) for convenience. A word over
$\Sigma^k$ is \emph{well-formed} if it is of the form $\overline{s} \, \overline{a}_n \dotsb \overline{a}_0 \, \overline{d} \, \overline{b}_0 \, \overline{b}_1 \dotsb$
with $s_i \in \{+,-\}$, $a_{i,j}, b_{i,j} \in \{0,1\}$ and $\overline{d}$ being the vector of $k$ dot-symbols. I.e.\ it 
starts with signs on all tracks, and each track contains exactly one dot, and these are all aligned. 
Such a word induces a $k$-tuple $(w_1,\ldots,w_k)$ of words over $\Sigma$ in the straightforward way: $w_i$ is represented by
the $\Sigma$-word $s_i a_{i,n-1} \ldots a_{i,0} . b_{i,0} b_{i,1} \ldots$ as above.
For example, let $k = 2$ and
\begin{displaymath}
	w = \twos{-}{+}\twos{0}{1}\twos{1}{0}\twos{1}{1}\twos{0}{0}\twos{1}{0}\twos{.}{.}\twos{0}{1}\Big(\twos{1}{0}\twos{1}{1}\Big)^\omega\ .
\end{displaymath}
It induces words $w_1, w_2$ that represent the numbers $\dec(w_1) = -13.5$ as well as $\dec(w_2) = 20 \frac{2}{3}$.
In the following, we will restrict our attention to well-formed words and write $\WF{\Sigma}{k}$ for the set of all such well-formed
$k$-track words. 
It is definable by a
\FGWNBA for any $k \ge 1$, namely the following one.
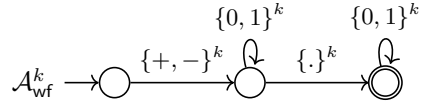
\begin{center}
\begin{tikzpicture}[semithick,initial text={\normalsize$\Aut{A}_{\mathsf{wf}}^k$}, every state/.style={inner sep=2pt, minimum size=4mm},font=\small,
                    node distance=18mm]
  \node[state,initial,initial where=left] (s0)               {};
  \node[state]                            (s2) [right of=s0] {};
  \node[state,accepting]                  (s3) [right of=s2] {};
  
  \path[->] (s0) edge              node [above] {$\{+,-\}^k$} (s2)
            (s2) edge [loop above] node [above] {$\{0,1\}^k$} ()
                 edge              node [above] {$\{.\}^k$}   (s3)
            (s3) edge [loop above] node [above] {$\{0,1\}^k$} ();
\end{tikzpicture}
\end{center}
Using that \FGWNBA are closed under intersection, w.l.o.g.\ all words are well-formed.
By the correspondence of a (well-formed) word from $(\Sigma^k)^\omega$ to $k$ words from $\Sigma^\omega$, we can view
the language of a \FGWNBA over the alphabet $\Sigma^k$ as a $k$-ary \emph{relation} of words $(w_1,\ldots,w_k)$ and, 
by the use of the decoding function $\dec$, as a $k$-ary relation of real numbers $(\dec(w_1),\ldots,\dec(w_k))$. We will 
therefore write $R(\Aut{A})$ instead of $L(\Aut{A})$ to denote the \emph{relation} of the automaton $\Aut{A}$ which,
technically, is just its language of the multi-track alphabet. 

We will need closure of the class of \FGWNBA-definable languages 
under several (arithmetical) operations which can be derived from two further basic ones: projections and products. 
Given a $k$-ary relation $R$ and a tuple $\pi = (i_1,\ldots,i_n)$ with $i_j \in \{1,\ldots,k\}$ for all $j$, the $\pi$-\emph{projection} 
of $R$ is the $n$-ary relation $\proj{(R)}{\pi} := \{ (w_{i_1},\ldots,w_{i_n}) \mid (w_1,\ldots,w_k) \in R \}$.
\begin{lemma}
\label{lem:projection}
	Let $\Aut{A}$ be a \FGWNBA s.t.\ $R(\Aut{A})$ is $k$-ary for some $k \geq 1$. Let $\pi \in \{1,\ldots,k\}^n$. There is a \FGWNBA 
	$\proj{(\Aut{A})}{\pi}$ of size $\mathcal{O}(|\Aut{A}|)$ s.t.\ $R(\proj{(\Aut{A})}{\pi}) = \proj{(R(\Aut{A}))}{\pi}$.
\end{lemma} 


Whilst, technically, the projection operation can be used to duplicate and re-arrange tracks in a multi-tracked word, we 
mostly use it to delete tracks. For example, if $R$ is a $3$-ary relation, then $\proj{R}{(1,3)}$ results from the collection 
of all tuples that are obtained by deleting the second component in a triple from $R$. 

Next, let $R_1$ be a $k_1$-ary and $R_2$ be a $k_2$-ary relation. The \emph{Cartesian product} is, as usual, the 
$(k_1+k_2)$-ary relation 
$R_1 \times R_2 := \{ (w_1,\ldots,w_{k_1},v_1,\ldots,v_{k_2}) \mid (w_1,\ldots,w_{k_1}) \in R_1, (v_1,\ldots,v_{k_2}) \in R_2 \}$.

\begin{lemma}
\label{lem:join}
\label{lem:product}
	Let $\Aut{A}_1,\Aut{A}_2$ be two \FGWNBA recognising a $k_1$-, resp.\ $k_2$-ary relation. There is a \FGWNBA 
	$\Aut{A}_1 \times \Aut{A}_2$ of size $\mathcal{O}(|\Aut{A}_1|\cdot |\Aut{A}_2|)$ s.t.\ 
	$R(\Aut{A}_1 \times \Aut{A}_2) = R(\Aut{A}_1) \times R(\Aut{A}_2)$.
\end{lemma} 


Let $1 \le i,j \le k$. It is easy to construct an automaton which accepts some $w \in \WF{\Sigma}{k}$ iff $w_i = w_j$, i.e.\ 
that checks for equality in the word representation of two numbers in a tuple. We need a more relaxed operation, namely an
automaton that accepts such a $k$-track word iff the $i$-th and $j$-th track represent the same number, possibly using 
different representations of it. Note for example that $+0.1^\omega$ and $+1.0^\omega$, or $+.0^\omega$ and $-.0^\omega$ 
represent the same number in each case. Luckily, these two examples already show all the possibilities to create different
representations of the same number in well-formed multi-track words, and these situations can be recognised by a \FGWNBA. 

\begin{lemma}
\label{lem:equality}
Let $k \ge 2$, $1 \le i < j \le k$. There is a \FGWNBA $\Aut{A}^k_{i=j}$ of size $\BigO{1}$ such that 
$R(\Aut{A}^k_{i=j}) = \{ w \in \WF{\Sigma}{k} \mid \dec(w_i) = \dec(w_j) \}$.
\end{lemma}


The automata $\Aut{A}^k_{i=j}$ are only used as auxiliary devices to form the closure of certain operations under different
number representation. As such, they are distinguished from other automata that we construct in the sense that most of them
operate on words of $k$ tracks which can be divided into $m$ \emph{input} tracks and $n$ \emph{output} tracks, s.t.\ $k= m+n$.
There is no technical difference between input and output tracks, though; the distinction is just useful in the specification
of certain operations. 

In such a setup it is natural to 
generalise the composition of two binary relations to ones of arbitrary, but matching arities. Suppose $R_1$ and $R_2$ are 
relations of arities $k_1$, resp.\ $k_2$, and $k \le \min \{k_1,k_2\}$. We regard $R_1$'s last $k$ tracks as its output and
$R_2$'s $k$ first tracks as its input. Then $R_1 \circ_k R_2 :=
\{ (u_1,\ldots,u_{k_1-k},w_{k_2-k+1},\ldots,w_{k_2}) \mid \exists v_1,\ldots,v_k \text{ s.t. } (u_1,\ldots,u_{k_1-k},v_1,\ldots,v_k) 
\in R_1, (v_1,\ldots,v_k,w_{k+1},\ldots,w_{k_2}) \in R_2 \}$.
We observe, for later constructions, that the class of \FGWNBA-definable languages is closed under such generalised
compositions.

\begin{lemma}
\label{lem:composition}
For $i \in \{1,2\}$ let $\Aut{A}_i$ be a \FGWNBA recognising a $k_i$-ary relation, and let $k \le \min \{k_1,k_2\}$. 
There is a \FGWNBA $\Aut{A}_1 \circ_k \Aut{A}_2$ of size $\mathcal{O}(|\Aut{A}_1|\cdot |\Aut{A}_2|)$ s.t.\ 
$R(\Aut{A}_1 \circ_k \Aut{A}_2) = R(\Aut{A}_1) \circ_k R(\Aut{A}_2)$.
\end{lemma}


We remark that Prop.~\ref{prop:emptiness} -- emptiness checks in time linear in the number of states -- is of course true for 
multi-track \FGWNBA as well. However, their alphabet $\Sigma^k$ is of size exponential in $k$, and this can lead to a number
of transitions that is exponential in $k$. However, on $n$ states there can be at most $n^2$ many different transitions
which calls for symbolic representations of $\Sigma^k$ in actual implementations. Also, the automata derived from DNN will be 
of exponential size in which case the possibly exponential size of the alphabet does not affect the statements made in the
following on asymptotic complexity. For the sake of simplicity, these are made with regards to the number of states of an automaton.

\section{Translating DNN into \FGWNBA}
\label{sec:translation}

The overall goal of this work is to develop the machinery that allows the input-output behaviour of a DNN to be captured 
by finite-state automata, here using \FGWNBA.
The definition of DNN given in Sec.~\ref{sec:prelim} implies an
inductive view on DNN: each DNN node itself is a DNN with one layer
consisting of one node, each DNN layer itself is a DNN with one layer and each
subset of consecutive layers is a DNN with several layers. 
We use this inductive view to first argue that there are \FGWNBA which capture the
computation of each node and then that there are \FGWNBA capturing whole layers
and complete DNN.

Let $v$ be a node computing $\relu(\sum_{i=1}^{k} w_i x_i + b)$. From its functional
form we can infer that the computation of $v$ is built from multiple instances 
of three fundamental operations: 1.\ multiplication of some arbitrary value with a fixed constant,
2.\ summation of arbitrary values and 3.\ the application of $\relu$ to some arbitrary value.
For each operation, we define a corresponding \FGWNBA and then combine these using the operations specified in 
Lemmas~\ref{lem:projection}, \ref{lem:join}, \ref{lem:composition} and the closure under $\cap$ and $\cup$.

\begin{lemma}
	\label{lem:help}
	Let $k \geq 2$,  $1 \leq i,j \leq k$ and $1 \leq i_1, \dotsc, i_n \leq k$ where $i \neq i_h$ and $i_h \neq i_l$ for $h,l \in \{1, \dotsc, n\}$. 
	There is a \FGWNBA 
	\begin{enumerate}
		\item $\Aut{A}^{k+1}_{i=\mathsf{add}(i_1,\ldots,i_n)}$ of size $2^{\BigO{k}}$such that 
		$R(\Aut{A}^{k+1}_{i=\mathsf{add}(i_1,\ldots,i_n)}) = \{ w \in \WF{\Sigma}{k+1} \mid \dec(w_{i}) = \sum_{h=1}^n \dec(w_{i_h}) \}$,
		\label{lem:help;add}
		\item $\Aut{A}^k_{j=\mathsf{relu}(i)}$ of size $\BigO{1}$ such that
		$R(\Aut{A}^k_{j=\mathsf{relu}(i)}) = \{ w \in \WF{\Sigma}{k} \mid \dec(w_j) = \relu(\dec(w_i))\}$,
		\label{lem:help;relu}
		\item $\Aut{A}^k_{j=\mathsf{mult}(c,i)}$ of size $2^{\BigO{\measure{c}}}$ such that
		$R(\Aut{A}^k_{j=\mathsf{mult}(c,i)}) = \{ w \in \WF{\Sigma}{k} \mid \dec(w_j) = c \cdot \dec(w_i) \}$ for every $c \in \Rats$ and 
		\label{lem:help;const_mult}
		\item $\Aut{A}^{k-1}_{i=\mathsf{const}(c)}$ of size $\BigO{2^{\measure{c}}}$ such that
		$R(\Aut{A}^k_{i=\mathsf{const}(c)}) = \{ w \in \WF{\Sigma}{k} \mid \dec(w_i) = c \}$.
		\label{lem:help;const}
	\end{enumerate}
\end{lemma}

Now, we lift these constructions to build the \FGWNBA $\Aut{A}_v$ representing the computation of a node $v$ in a DNN. 

\begin{lemma}
	\label{sec:translation;lem:node}
	\label{lem:node}
	Let $k \ge 2$,  $h < j \le k$, and $v$ be a DNN-node computing $ \relu(b + \sum_{i=1}^h c_i x_i)$. 
	There is a \FGWNBA $\Aut{A}^k_{j=v(1, \dotsc, h)}$ of size $2^{\BigO{\measure{v}}}$ s.t.\ 
	$R(\Aut{A}^k_{j=v(1, \dotsc, h)}) = \{w \in \WF{\Sigma}{k} \mid \dec(w_{j}) =  \relu(b + \sum_{i=1}^h c_i \cdot \dec(w_i))\}$.
\end{lemma}

\begin{proof}
Note that $\Aut{A}^k_{j=v(1, \dotsc, h)}$ is supposed to work over $k$-track words in which the first $h$ tracks contain inputs $x_1,\ldots,x_h$, and the node's output is expected in the $(j-h)$-th output track which is the $j$-th overall. 
Let $\Aut{A}^k_{j=v(1, \dotsc, k)}$ be equal to
$\big((\bigcap_{i=1}^h  \Aut{A}^{g+2}_{k+i=\mathsf{mult}(c_i,i)}) \cap 
\Aut{A}^{g+2}_{g+1=\mathsf{const}(b)} \cap \Aut{A}^{g+2}_{g+2=\mathsf{add}(k+1,\ldots,g+1)} \cap
\proj{\Aut{A}^{g+2}_{j=\mathsf{relu}(g+2)}\big)}{1,\ldots,k}$.
where $g = k+h$. It uses $h+2$ additional and intermediate tracks that hold, respectively, for input values
$x_1,\ldots,x_h$ in the first $h$ tracks, the values $c_1 \cdot x_1,\ldots,c_h\cdot x_h$, the bias $b$, and their sum. By also 
insisting that the $j$-th track holds the ReLU-value of that sum, we model exactly the node's computation. 
Since it is constructed using $h+2$ intersections of \FGWNBA of size that is either constant (for the addition) or
of size bounded by the involved rational constants (for the multiplication and the bias), the overall size can be
estimated as $2^{\BigO{\measure{v}}}$. \qed 
\end{proof}

Using the inductive view on DNN described above, we are now set to provide the translation
of DNN into input-output-equivalent \FGWNBA.

\begin{theorem}
	\label{th:dnn}
	\label{thm:dnn}
	Let $N$ be a DNN with input dimension $m$ and output dimension $n$. There is a \FGWNBA
	$\Aut{A}_N$ of size $2^{\BigO{\measure{N}}}$ s.t.\ 
	$R(\Aut{A}_N) = \{w \in \WF{\Sigma}{m+n} \mid N(\dec(w_1), \ldots, \dec(w_m)) =
	(\dec(w_{m+1}), \ldots, \dec(w_{m+n}))\}$.
\end{theorem}
\begin{proof}
	Assume that $N$ has $k$ layers $l_1,\ldots,l_k$. For each layer $l_i$, we construct an \FGWNBA $\Aut{A}_i$ recognising the 
	relation between inputs to this layer and immediate outputs computed by it. 
	Take a layer $l_i = (v^i_1, \ldots, v^i_{n_i})$, and assume that it takes $m_i$ inputs (which must also be the number
	of outputs of the previous layer). Obviously, it produces $n_i$ outputs as this is the number of nodes in this layer.
	Moreover, we have $m_1 = m$ and $n_k = n$, i.e.\ the inputs to the DNN are inputs to the first layer, and the outputs of the
	last layer are the outputs of the DNN. 
	The desired \FGWNBA can be obtained from the \FGWNBA for the nodes $v^i_j$ of this layer according to Lemma~\ref{lem:node} 
	as $\Aut{A}_i := \bigcap_{j=1}^{n_i} \Aut{A}^{m_i+n_i}_{m_i+j=v^i_j(1, \dotsc, m_i)}$. This produces a \FGWNBA with $m_i$ inputs and $n_i$ outputs 
	which contains, in the $j$-th output track, the result of the computation done by the $j$-th node in this layer on the
	inputs contained in the $m_i$ input tracks.
	Finally, a \FGWNBA for the relation computed by the DNN $N$ is then obtained simply as 
	$\Aut{A}_N := \Aut{A}_1 \circ_{n_1} \ldots \circ_{n_{k-1}} \Aut{A}_k$. Note that relation composition is in fact 
	associative. 
	The size of $\Aut{A}_N$ can be bounded by $2^{\BigO{\measure{N}}}$ because of the following observation: for a layer of 
	$n$ nodes $v$ we need to form a product of $n$ automata, each of size bounded by $2^{\BigO{\measure{v}}}$, i.e.\ we get
	a size of $2^{\BigO{n \cdot \measure{v}}}$ whose exponential corresponds to the size that a layer requires in a DNN
	representation. Likewise, forming the composition for $k$ layers results in an \FGWNBA of size bounded by 
	$2^{\BigO{k\cdot n \cdot \measure{v}}} = 2^{\BigO{\measure{N}}}$ where $k \cdot n$ is an upper bound for the
	number of nodes. \qed
\end{proof}

\section{Use Cases: Analysing DNN Using \FGWNBA}
\label{sec:verification}

We consider two topics -- formal verification and interpretation of DNN.
The former is concerned with different safety properties, among which
adversarial robustness and output reachability guarantees belong to the most important ones. 
Interpretation of DNN is concerned with techniques generating 
human-understandable explanations for the behaviour of DNN, for example,
an explanation why a DNN computes some specific output given some input.

\subsection{Verifying DNN Using \FGWNBA}

\subsubsection*{Adversarial robustness.} This is exclusively concerned with \emph{classifier DNN}.
A classifier DNN is used to assign to a given input one of the \emph{classes} $\{c_1, \dotsc, c_k\}$. 
Typically, such a classifier $N$ is built as follows: $N$ consists of a DNN $N'$ with output dimension $k$ and 
an additional \emph{softmax layer}. It consumes the output $(y_1, \dotsc, y_k)$ of $N'$ and 
computes a probability for each class $c_i$. The input $\bs{x}$ is then said to be classified into $c_i$ if 
its probability is maximal. However, the actual assigned class can be directly inferred from the output of $N'$ by 
assigning class $c_j$ to $\bs{x}$ such that $j$ is a maximal output dimension.
A formal definition of adversarial robustness relies on a distance measure on real vectors. The one that is commonly
used is the one induced by the $1$-norm of vectors \cite{HKRSSTWY20_survey}. 
Let $\bs{r} \in \Reals^m$. Its \emph{$1$-norm} is $||\bs{r}||_1 = \sum^m_{i=1} |r_i|$.
It induces the \emph{Manhattan distance} of $\bs{r}$ and some $\bs{s} \in \Reals^m$, defined as 
$||\bs{r}, \bs{s}||_1 = \sum^m_{i=1} |r_i - s_i|$. 
The \emph{$d$-neighbourhood} with $d \in \Reals^{\geq 0}$ of $\bs{r}$ is defined as the set 
$\{\bs{s} \in \Reals^m \mid ||\bs{r}, \bs{s}||_1 \leq d\}$.
Let $N$ be a classifier with input dimension $m$ and output dimension $n$ and $1 \leq h \leq n$.
We call a triple $P=(\bs{r}, d, h)$ with $\bs{r} \in \Rats^m, d \in \Rats$ an
\emph{adversarial robustness property} (ARP) and say that $N$ satisfies $P$, written $N \models P$, if $N(\bs{r}')_h > N(\bs{r'})_{h'}$ for all $h' \ne h$ and
all $\bs{r}'$ in the $d$-neighbourhood of $\bs{r}$. 
In other words, the entire $d$-neighbourhood is classified as belonging to class $c_h$.
We measure the size of $P$ by $\measure{P} = \measure{\bs{r}} + \measure{d} + \measure{h}$ where
$\measure{\bs{r}}$ is the sum of the measure of its elements $\measure{r_i}$.

The construction of input-output equivalent \FGWNBA, established in Sect.~\ref{sec:translation}, can be used to verify ARPs. 
For a DNN $N$ and and ARP $P$ we combine three \FGWNBA $\Aut{A}^P_{\mathsf{in}}, \Aut{A}^P_{\mathsf{out}}$ and 
$\Aut{A}_N$  to a \FGWNBA that can be used to check whether $N \models P$ holds. Based on the explanations above we disregard the $\mathrm{softmax}$ layer 
of $N$, giving a usual DNN. Then $\Aut{A}_N$ is defined by Thm.~\ref{th:dnn}.  
The other two automata accept only the valid input, respectively output vectors according to $P$. 

The automaton $\Aut{A}^P_{\mathsf{in}}$ should accept words corresponding to vectors $\bs{x}$ that are included in the 
$d$-neighbourhood of $\bs{r}$, i.e.\ for which $\sum_{j=1}^{m} |r_j +  (-1\cdot x_j)| \leq d$ holds. 
From Sect.~\ref{sec:translation} we know that there are automata recognising the operations of addition and multiplication
by the constant $-1$. 


\begin{lemma}
\label{lem:arp_in}
	Let $k \ge 2$, $1 \le m < i \le k$, $\bs{r} \in \Rats^m$, $d \in \Rats$ with $d > 0$ and $P=(\bs{r},d,i)$. There is a \FGWNBA
	$\Aut{A}^{k,P}_{\mathsf{in}}$ of size $2^{\BigO{\measure{P}}}$ s.t.\ for all $w \in \WF{\Sigma}{k}$ we have: 
	$w \in R(\Aut{A}^{k,P}_{\mathsf{in}})$ iff $\sum_{i=1}^{m} |r_i +  (-1\cdot \dec(w_i))| \leq d$.
\end{lemma} 

Given this construction, Theorem \ref{thm:dnn} and the fact that \FGWNBA are closed under intersection, 
we get that for each ARP there is a \FGWNBA which recognises its validity.
\begin{theorem}
\label{thm:arp}
Let $N$ be a DNN with $m$ inputs and $n$ outputs, and $P = (\bs{r},d,i)$ be an ARP with $1 \le i \le n$. There is a \FGWNBA 
$\Aut{A}^{N,P}_{\mathsf{arp}}$ of size $2^{\BigO{\measure{N} + \measure{P}}}$ s.t.\ 
$R(\Aut{A}^{N,P}_{\mathsf{arp}}) = \emptyset$ iff $N \models P$. 
\end{theorem}
\begin{proof}
Let $k := m+n$. Note that $\Aut{A}_N$ is a $k$-track \FGWNBA recognising the input-output behaviour of $N$. Let 
$\Aut{A}^{N,P}_{\mathsf{arp}} := \Aut{A}_N \cap \Aut{A}^{k,P}_{\mathsf{in}} \cap \overline{\Aut{A}}^{k,P}_{\mathsf{out}}$ where
$\overline{\Aut{A}}^{k,P}_{\mathsf{out}} := (\bigcup_{i'=1}^{i-1} \Aut{A}^k_{m+i \le m+i'}) \cup (\bigcup_{i'=i+1}^{n} \Aut{A}^k_{m+i \le m+i'})$ accepts a word with 
$n$ output tracks if the number encoded in the $i$-th output is not greater than those in any other output track. 
The size claim is a straightforward result from the sizes of the subautomata 
$\Aut{A}_N$, $\Aut{A}^{k,P}_{\mathsf{in}}$ and $\overline{\Aut{A}}^{k,P}_{\mathsf{out}}$.
Consequently, $\Aut{A}^{N,P}_{\mathsf{arp}}$ accepts a $k$-track word if it encodes a vector $\bs{x} \in \Reals^m$ in its
first $m$ tracks that is within the $d$-neighborhood of $\bs{r}$, s.t. the following $n$ tracks encode $N(\bs{x})$ and 
their $i$-th component is not strictly maximal. This is the case if and only if $N \not\models P$. \qed
\end{proof}

\subsubsection*{Output reachability.}
This is used to certify that specific ``misbehaviour'' of DNN does not occur. 
A formal definition hinges on a notion of valid inputs and outputs. Commonly, this is 
done using specifications defining (convex) sets of real vectors. 
A \emph{vector specification} $\varphi$ over variables $x_1,\ldots,x_k$ is a conjunction $\varphi$ of statements of the form 
$(\sum_{i=1}^k c_i \cdot x_i) \le b$ where $c_i,b \in \Rats$. 
Let $\bs{r} = (r_1, \dotsc, r_k) \in \Reals^k$. We say that $\bs{r}$ satisfies $\varphi$ if each 
inequality $t \leq b$ in $\varphi$ is satisfied in real arithmetic when each $x_i$ is given the value $r_i$. 
Let $N$ be a DNN with input dimension $m$ and output dimension $n$,
let $\varphi_{\text{in}}$ be a vector specification over $x_1, \dotsc, x_m$ and let $\varphi_{\text{out}}$ be a vector specification over 
$y_1, \dotsc, y_n$. We call the tuple $P = (\varphi_{\text{in}}, \varphi_{\text{out}})$ an \emph{output reachability property (ORP)} 
and say that $N$ satisfies $(\varphi_{\text{in}}, \varphi_{\text{out}})$, written $N \models P$, if there is $\bs{r} \in \Reals^m$ s.t.\ 
$\bs{r} \models \varphi_{\text{in}}$ and $N(\bs{r}) \models \varphi_{\text{out}}$. We define the size of $P$ by $\measure{P} =
\measure{\varphi_\text{in}} + \measure{\varphi_\text{out}}$ where the measure of a specification $\varphi$ is the sum of the measures of parameters $c_i, b$ ocurring in some inequality.
\begin{theorem}
\label{thm:orp}
Let $N$ be a DNN with $m$ inputs and $n$ outputs, and $P = (\varphi_{\text{in}},\varphi_{\text{out}})$ be an ORP. There is a 
\FGWNBA $\Aut{A}^{N,P}_{\mathsf{orp}}$ of size $2^{\BigO{\measure{N} + \measure{P}}}$ s.t.\ 
$R(\Aut{A}^{N,P}_{\mathsf{orp}}) = \emptyset$ iff $N \models P$. 
\end{theorem}

\begin{proof}
Similar to the constructions in Thm.~\ref{thm:arp}, one can build, given $k$ and 
a linear inequality $\psi = \sum_{i=1}^k c_i \cdot x_i \le b$ with rational constants, a \FGWNBA $\Aut{A}^\psi_{\mathsf{in}}$
that accepts a well-formed $k$-track word iff the first $k$ tracks encode numbers $x_1,\ldots,x_k$ that satisfy $\psi$. Likewise,
we can build such a \FGWNBA $\Aut{A}^\psi_{\mathsf{out}}$ that does the same for the last $n$ tracks. Note that the size of
these automata is exponential in the measure of the parameters $c_i,b$.
Then we get that $\Aut{A}^{N,P}_{\mathsf{orp}} := (\bigcap_{\psi \in \varphi_{\text{in}}} \Aut{A}^\psi_{\text{in}}) \cap \Aut{A}_N \cap
(\bigcap_{\psi \in \varphi_{\text{out}}} \Aut{A}^\psi_{\text{out}})$
accepts a word iff it encodes some $\bs{x}$ satisfying $\varphi_{\text{in}}$, s.t.\ $N(\bs{x})$ satisfies $\varphi_{\text{out}}$.
The size claim about $\Aut{A}^{N,P}_{\mathsf{orp}}$ is a straightforward result from the intersection and 
the size of $\Aut{A}_N$ and the specification automata. \qed
\end{proof}

\subsection{Interpreting DNN with \FGWNBA}

Zhang et al.\ \cite{ZhangTLT21_surveyinterp} present a
three-dimensional taxonomy for interpretation techniques: 
\emph{post-hoc} or \emph{ad-hoc} interpretation either generates explanations 
for common neural network models or focuses on constructing 
neural network models that improve interpretability. 
\emph{Examples, attribution, hidden semantics} or \emph{rules} characterise 
the type of explanation. \emph{Global, local} or \emph{semi-local} explanations
concern the model's overall behaviour, that of a single input value, resp.\ something in between.
In the following, we introduce a widely considered post-hoc, attribution and local interpretation approach.

We start with an example. Assume some image-classification task, for instance the task to distinguish pictures of dogs and pigs. 
Commonly, such a task is addressed using a model like \emph{Convolutional Neural Networks} \cite{KhanSZQ20_cnnsurvey} which processes a picture by
computing layer-by-layer higher-order features of the picture
and then classifies it based on these. A natural explanation for
the CNN's decision is the image's regions that the CNN focuses on for 
making its decision, classifying it as either a picture of a dog or of a pig.
For example, we would gain confidence in the CNN decision 
if we can prove that it focuses on the form of the snout of the animal (lengthy vs.\ flat) or the texture of its
outer contours (fluffy vs.\ smooth). In technical terms, the task
is to find the most important input dimensions, i.e.\ pixels of the image, that determine 
the output of the CNN. A widely used interpretation technique addressing this problem is called Integrated 
Gradient \cite{SundararajanTY17_intgrad}.
In the context of our general DNN model, we formulate the task of finding the most 
important features of an input as a decision problem: 
given a DNN $N$, some input $\bs{r} \in \Rats^m$ and $I \subseteq \{1,\ldots,m\}$, decide whether
for every $\bs{x} \in \Rats^k$ which equals $\bs{r}$ on the dimensions in $I$ we have $N(\bs{x}) = N(\bs{r})$.
In correspondence to \cite{BarceloM0S20_interprcompl}, we call
this problem \textsc{MSR} (for \emph{minimum sufficient reason}). 
An instance is of the form $P = (\bs{r},I)$. As before, we write $N \models P$
to indicate that $N$ satisfies the instance $P$. The measure $\measure{P}$ is given by the sum
of the measures of $\measure{r}$ and $\measure{I}$, which are defined in the obvious way.

\begin{theorem}
\label{thm:minsuffreason}
Let $N$ be a DNN with input dimension $m$ and output dimension $n$, and some $P = (\bs{r},I)$ 
with $\bs{r} \in \Rats^m$ 
and $I \subseteq \{1,\ldots,m\}$. There is a \FGWNBA $\Aut{A}^{N,P}_{\mathsf{msr}}$ of size $2^{\BigO{\measure{N} + \measure{P}}}$
s.t.\ $R(\Aut{A}^{N,P}_{\mathsf{msr}}) = \emptyset$ iff $N \models P$.
\end{theorem}

\begin{proof}
We need an auxiliary automaton $\Aut{A}^k_{i\ne j}$ which accepts a $k$-track word iff its $i$-th and $j$-th track
represent different numbers. Note that we cannot simply complement $\Aut{A}^k_{i=j}$ from Lemma~\ref{lem:equality} since
weak NBA, let alone \FGWNBA, are not closed under complement. However, it is easy to construct $\Aut{A}^k_{i\ne j}$ directly,
or as $\Aut{A}^k_{i < j} \cup \Aut{A}^k_{j < i}$ using the \FGWNBA from Lemma~\ref{lem:smaller}.
Let $\Aut{A}_N$ be the $l$-track \FGWNBA with $k=m+n$ recognising $N$'s input-output relation from Thm.~\ref{thm:dnn}. 
To construct $\Aut{A}^{N,P}_{\mathsf{msr}}$ we use two copies that work in parallel, computing $N$'s output on some 
$\bs{x}$ and on $\bs{r}$, checking whether the inputs agree on the dimensions in $I$ and whether their outputs disagree. 
Define $\Aut{A}^{N,P}_{\mathsf{msr}}$ as 
$
(\Aut{A}_N \bowtie \Aut{A}^l_{\mathsf{wf}}) \cap (\Aut{A}^l_{\mathsf{wf}} \bowtie \Aut{A}_N) \cap
(\bigcap_{i=1}^m \Aut{A}^{2l}_{l+i=\mathsf{const}(r_i)}) \cap
(\bigcap_{i\in I} \Aut{A}^l_{l+i=i}) \cap 
(\bigcup_{i=1}^n \Aut{A}^{2l}_{m+i \ne l+m+i})$.
The size of this automaton is determined by the intersection of $\Aut{A}_N$ and automata $\Aut{A}^{2l}_{l+i=\mathsf{const}(r_i)}$ parametrized by parts of $I$.
Take a $2l$-track word $w$. Define $\bs{x}$ as $(\dec(w_1),\ldots,\dec(w_m))$ and define $\bs{x'}$ as $(\dec(w_{l+1}),\ldots,$ $\dec(w_{l+m}))$. Define $\bs{y}$ as $(\dec(w_{m+1}),\ldots,\dec(w_{m+n}))$ and define $\bs{y'}$ as $(\dec(w_{l+m+1},\ldots,\dec(w_{2l}))$. Then 
$w \in R(\Aut{A}^{N,P}_{\mathsf{msr}})$ iff $\bs{x'} = \bs{r}$, 
$x_i = r_i$ for all $i \in I$, and $N(\bs{x}) = \bs{y} \ne \bs{y'} = N(\bs{x'}) = N(\bs{r})$, i.e.\ if and only if 
$\bs{x}$ witnesses the fact that $N \not\models P$. \qed
\end{proof}

In practice, the dimensions in $I$ are not explicitly given, but only a number $l \le m$ is given with the proviso that
a set $I$ of input dimensions should be found s.t.\ $|I|=l$ and this set provides a minimum sufficient reason for the
classification of $\bs{r}$. Clearly, by invoking Thm.~\ref{thm:minsuffreason}, at most $\binom{m}{l}$ 
times a counterexample can be found using successive emptiness checks. It remains to be seen whether this can be improved,
for instance by not enumerating all sets $I$ in a brute-force way but to construct one from smaller ones for instance.

\section{Discussion and Outlook}
\label{sec:outlook}

We presented an automata-theoretic framework that can be used to address a broad range of analysis tasks on DNN. The core 
result (Thm.~\ref{thm:dnn}) tranforms a DNN $N$ into an eventually-always weak B\"uchi automaton of exponential size that 
exactly captures (word encodings) of the input-output pairs defined by $N$. Our key observations (Thms.~\ref{thm:arp},
\ref{thm:orp} and \ref{thm:minsuffreason}) are that different particular verification and interpretation problems can
be reduced to emptiness checks for these automata.

The approach presented in Sect. \ref{sec:translation} is conceptual rather than practical. In order
to obtain practically useful automata-based tools for DNN analysis, further work is needed.
 
The exposition here is done w.r.t.\ to a particular neural network model. Hence, further work  consists
of identifying other classes of NN which can be translated similarly into finite-state automata, including special cases that
lead to more efficient translations. 
It also remains to be seen whether the tools presented here can be used meaningfully in the analysis of DNN subcomponents, l
ike a subset of subsequent layers, of a DNN only. 
Layerwise verfication procedures, like interval propagation \cite{LiLYCHZ19_intervalprop} for instance, 
have been shown to be useful in DNN verification in general. 
The use of NFA and finite words, instead of \FGWNBA and infinite words, 
constitutes an abstraction of a DNN's behaviour in the form of a function of type $\Reals^m \to \Reals^n$ to functions on
some subset. For instance, when cutting down all \FGWNBA to accept immediately rather than read dot symbols, we would obtain
NFA over $\{+,-,0,1\}$ that approximate a DNN's behaviour as a function of type $\Ints^m \to \Ints^n$. We aim to investigate 
this idea more formally, making use of a well-developed theory of abstraction and refinement 
\cite{CoCo:abstract,journals/jacm/ClarkeGJLV03}, with the aim of acquiring a better understanding of the possibilities to trade 
precision for efficiency in DNN analysis.


Besides, future research should focus on the identification of analysis problems for which the automata-theoretic framework is
genuinely superior compared to other techniques, as one obtains, in the form of the automaton $\Aut{A}_N$, a finite representation 
of the entire input-output behaviour of $N$. This may include transferring the comparison of two DNN $N_1$ and $N_2$ to their
respective automata representations $\Aut{A}_{N_1}$ and $\Aut{A}_{N_2}$. We can compare the behaviour of $N_1$ and $N_2$ by 
investigating, for instance, the intersection of $R(A_{N_1})$ and $R(A_{N_2})$, or their symmetric difference (which will in 
general only be definable by an NBA rather than a \FGWNBA), to obtain notions of diverging behaviour or of equivalence between DNN.

\subsubsection*{Acknowledgements.}
We would like to thank R\"udiger Ehlers for fruitful discussions on this topic and helpful comments on
an earlier draft of the paper.

\bibliographystyle{splncs04}
\bibliography{refs}

\newpage
\appendix 
\section{Proofs}
\label{sec:proofs}


\subsubsection*{Lemma \ref{lem:projection}.}
This is a standard construction that applies the projection pointwise to the tuples in each transition: if $\pi = (i_1,\ldots,i_n)$
then $\proj{(\Aut{A})}{\pi}$ is obtained from $\Aut{A}$ by replacing each transition $(q,[a_1,\ldots,a_k],p)$ by
$(q,[a_{i_1},\ldots,a_{i_n}],p)$. As a consequence, an accepting run of $\Aut{A}$ on $(w_1,\ldots,w_k) \in (\Sigma^k)^\omega$ 
is an accepting run of $\proj{(\Aut{A})}{\pi}$ on $(w_{i_1},\ldots,w_{i_n}) \in (\Sigma^n)^\omega$ and vice-versa. \qed

\subsubsection*{Lemma \ref{lem:join}.}
We use a product construction again. Let $\Aut{A}_i = (Q_i,\Sigma^{k_i},$ $q_0^i,\delta_i,F_i)$
for $i \in \{1,2\}$. Define $\Aut{A}_1 \times \Aut{A}_2$ as 
$(Q_1 \times Q_2, \Sigma^{k_1+k_2}, (q_0^1,q_0^2), \delta, F_1 \times F_2)$ with
\begin{align*}
	&\big((q_1,q_2),[a_1,\ldots,a_{k_1},b_1,\ldots,b_{k_2}],(p_1,p_2)\big) \in \delta \quad \text{iff} \\
	&\quad (q_1,[a_1,\ldots,a_{k_1}],p_1) \in \delta_1 \enspace \text{and} \enspace (q_1,[b_1,\ldots,b_{k_2}],p_2) \in \delta_2\ .
\end{align*}
Consequently, when $q^1_0,q^1_1,\ldots$ and $q^2_0,q^2_1,\ldots$ are accepting runs of $\Aut{A}_1$ on $(w_1,\ldots,$ $w_{k_1})$,
and of $\Aut{A}_2$ on $(v_1,\ldots,v_{k_2})$ respectively, then $(q^1_0,q^2_0),(q^1_1,q^2_1),\ldots$ 
is an accepting run of $\Aut{A}_1 \times \Aut{A}_2$ on $(w_1,\ldots,w_{k_1},v_1,\ldots,v_{k_2})$ and vice-versa. \qed

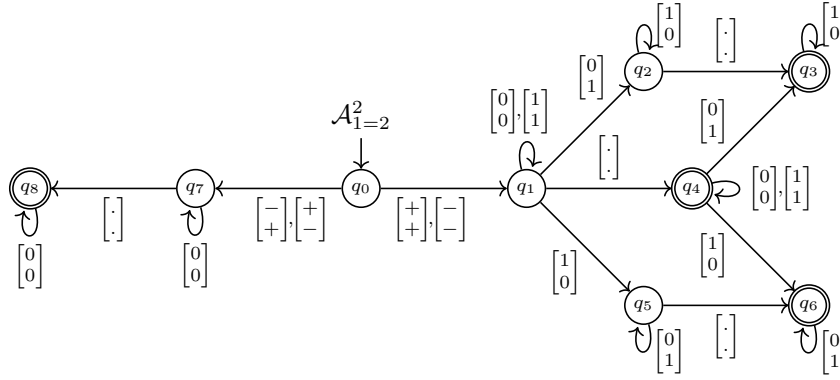
\begin{figure}[t]
	\begin{center}
		\begin{tikzpicture}[semithick,initial text={\normalsize$\Aut{A}^2_{1=2}$}, every state/.style={inner sep=2pt, minimum size=4mm},font=\scriptsize,
			node distance=22mm]
			\node[state,initial,initial where=above] (s0)                     {$q_0$};
			\node[state]                             (s2) [right of=s0]       {$q_1$};
			\node[state]                             (u1) [above right of=s2] {$q_2$};
			\node[state,accepting]                   (u2) [right of=u1] {$q_3$};
			\node[state]                             (d1) [below right of=s2] {$q_5$};
			\node[state,accepting]                   (d2) [right of=d1] {$q_6$};
			\node[state,accepting]                   (s3) [right of=s2]       {$q_4$};
			\node[state]                             (n2) [left of=s0]  {$q_7$};
			\node[state,accepting]                   (n3) [left of=n2]        {$q_8$};
			
			\path[->] (s0) edge              node [below,pos=.4]            {$\twos{+}{+}\!\!,\!\!\twos{-}{-}$} (s2)
			edge              node [below,pos=.4]        {$\twos{-}{+}\!\!,\!\!\twos{+}{-}$} (n2)
			(s2) edge [loop above] node [above,xshift=-1mm]                  {$\twos{0}{0}\!\!,\!\!\twos{1}{1}$} ()
			edge              node [above left=-1mm,pos=.8] {$\twos{0}{1}$} (u1)
			edge              node [below left=-1mm] {$\twos{1}{0}$} (d1)
			edge              node [above]                  {$\twos{.}{.}$} (s3)
			(s3) edge [loop right] node [right]                  {$\twos{0}{0}\!\!,\!\!\twos{1}{1}$} ()
			edge              node [above left=-1mm,pos=.3]       {$\twos{0}{1}$} (u2)
			edge              node [below left=-1mm,pos=.3]       {$\twos{1}{0}$} (d2)
			(u1) edge              node [above]                  {$\twos{.}{.}$} (u2)
			edge [loop above] node [right]                  {$\twos{1}{0}$} ()
			(u2) edge [loop above] node [right]                  {$\twos{1}{0}$} ()               
			(d1) edge              node [below]                  {$\twos{.}{.}$} (d2)
			edge [loop below] node [right]                  {$\twos{0}{1}$} ()
			(d2) edge [loop below] node [right]                  {$\twos{0}{1}$} ()               
			(n2) edge [loop below] node [below]                  {$\twos{0}{0}$} ()
			edge              node [below]                  {$\twos{.}{.}$} (n3)
			(n3) edge [loop below] node [below]                  {$\twos{0}{0}$} ();
		\end{tikzpicture}
	\end{center}
	\caption{\FGWNBA that recognises the binary equality relation.}
	\label{fig:equality}
\end{figure}
\subsubsection*{Lemma \ref{lem:equality}.}
$\Aut{A}^2_{1=2}$ is shown in Fig.~\ref{fig:equality}. Note that any $w \in \WF{\Sigma}{2}$ with $w_1 = w_2$ is accepted
via the run that simply moves horizontally right from the initial state. The transitions leading up or down are used to capture
encodings of the same number ending in $10^\omega$, resp.\ $01^\omega$ .The runs leading to the left from the initial state cover the special case of the positive and
negative representation of $0$. To obtain $\Aut{A}^k_{i=j}$ for arbitrary $k,i,j$, it suffices to extend the transition labels
in the automaton of Fig.~\ref{fig:equality} to contain arbitrary bits in positions other than $i$ and $j$. \qed

\subsubsection*{Lemma \ref{lem:composition}.}
$\Aut{A}_1 \circ_k \Aut{A}_2$ can be obtained as
\begin{displaymath}
	\proj{\big( (\Aut{A}_1 \times \Aut{A}_{\mathsf{wf}}^{k_2}) \cap (\Aut{A}_{\mathsf{wf}}^{k_1} \times \Aut{A}_2) \cap 
		\bigcap\limits_{i=1}^{k} \Aut{A}^{k_1+k_2}_{k_1-k+i=k_1+i}\big)}{1,\ldots,k_1-k,k_1+k+1,\ldots,k_1+k_2} \ .
\end{displaymath}
The inner part checks simultaneously whether a given $(k_1+k_2)$-track word is s.t.\ the first $k_1$ tracks are recognised
by $\Aut{A}_1$, the last $k_2$ tracks are recognised by $\Aut{A}_2$, and that the output tracks of the former represent the same
numbers as the input tracks of the latter. The projection is then used to delete all tracks apart from the inputs of the former
and the outputs of the latter. \qed

\subsection{Translating DNN into \FGWNBA}

\begin{figure}[t]
	\begin{center}
		\begin{tikzpicture}[semithick,initial text={\normalsize$\Aut{H}_{3=1{+}2}$}, every state/.style={inner sep=2pt, minimum size=4mm},font=\tiny,
			node distance=22mm]
			\node[state]                             (u2)                     {};
			\node[state]                             (u3) [right of=u2]       {};
			\node[state,accepting]                   (d2) [right of=u3]       {};
			\node[state,accepting]                   (d3) [right of=d2]       {};
			\node[state,initial,initial where=above] (i)  [left of=u2]        {};
			\node                                    (g0) [below of=i, node distance=16mm]        {$\cdots$};
			\node                                    (g1) [left of=i, node distance=20mm]         {$\cdots$};
			
			\path[->] (i)  edge              node [above,pos=.5] {$\threes{+}{+}{+}\!\!,\!\!\threes{-}{-}{-}$}       (u2)
			edge [shorten >=1mm]             node [left]          {$\threes{+}{-}{+}\!\!,\!\!\threes{-}{+}{-}$}       (g0)
			edge [shorten >=1mm]             node [above]          {$\threes{-}{+}{-}\!\!,\!\!\threes{+}{-}{+}$}       (g1)
			(u2) edge [loop above] node [above] {$\threes{0}{0}{0}\!\!,\!\!\threes{0}{1}{1}\!\!,\!\!\threes{1}{0}{1}$} ()
			edge [bend left=40]  node [above] {$\threes{0}{0}{1}$}                                                   (u3)
			(u3) edge [loop below] node [below] {$\threes{0}{1}{0}\!\!,\!\!\threes{1}{0}{0}\!\!,\!\!\threes{1}{1}{1}$} ()
			edge [bend left=40]  node [above] {$\threes{1}{1}{0}$}                                                   (u2)
			(d2) edge [loop above] node [above] {$\threes{0}{0}{0}\!\!,\!\!\threes{0}{1}{1}\!\!,\!\!\threes{1}{0}{1}$} ()
			edge [bend left=40]  node [above] {$\threes{0}{0}{1}$}                                                   (d3)
			(d3) edge [loop below] node [below] {$\threes{0}{1}{0}\!\!,\!\!\threes{1}{0}{0}\!\!,\!\!\threes{1}{1}{1}$} ()
			edge [bend left=40]  node [above] {$\threes{1}{1}{0}$}                                                   (d2);
			
			\path[->,draw,rounded corners] (u2) |- node [left,pos=.25] {$\threes{.}{.}{.}$} ++(2,-1.5) -| (d2); 
			\path[->,draw,rounded corners] (u3) |- node [right,pos=.25] {$\threes{.}{.}{.}$} ++(2,1.5) -| (d3); 
		\end{tikzpicture}
	\end{center}
	\caption{\FGWNBA that recognises a subset of the ternary addition relation.}
	\label{fig:addition}
\end{figure}
\subsubsection*{Lemma \ref{lem:help}.\ref{lem:help;add}.}
Fig.~\ref{fig:addition} shows a sketch of an auxiliary automaton $\Aut{H}_{3=\mathsf{add}(1,2)}$ which almost recognises the ternary addition 
relation $R^3_{\mathsf{add}} := \{ w \in \WF{\Sigma}{3} \mid \dec(w_3) = \dec(w_1) + \dec(w_2) \}$. 
The automaton $\Aut{H}_{3=\mathsf{add}(1,2)}$ is obtained 
as the disjoint union of three independent components that a run enters depending on the
first input symbol $[s_1,s_2,s_3]$ representing the signs of the two summands and the sum. The component handling the case of  
$[+,+,+]$ (and $[-,-,-]$ by symmetry because $x_1 + x_2 = x_3$ iff $-x_1 - x_2 = -x_3$) is shown on the right. It contains the typical 
construction for bitwise addition from highest bits to lower ones: the automaton remembers whether the next read symbol 
$(a_1,a_2,b)$ is supposed to produce a carry bit (1) or not (0). It also checks, as usual, that eventually $[.,.,.]$ is being read.

The other two components are not shown in Fig.~\ref{fig:addition} for the purpose of brevity. Both are structurally
very similar to the 4-state component just described. For example, the component entered upon reading $[+,-,+]$ or, by symmetry,
$[-,+,-]$, can be obtained from this one simply by uniformly swapping components in the letters so that $[a_1,a_2,b]$ becomes
$[a_2,b,a_1]$. Note that $x_1 - x_2 = x_3$ iff $x_2 + x_3 = x_1$. Hence, by swapping the roles of $x_1$ (which is now the sum) and 
$x_3$ (which is now a summand) one can re-use the mechanisms of adding two positive summands in this case. Likewise, the component
entered upon reading $[+,-,-]$ or $[-,+,+]$ is obtained from the one described first by swapping the symbols in the second and 
third track.

$\Aut{H}_{3=\mathsf{add}(1,2)}$ only recognises a subset of $R^3_{\mathsf{add}}$. For
example, 
\begin{displaymath}
	\threes{+}{+}{+} \threes{0}{0}{1} \threes{0}{0}{0} \threes{1}{1}{0} \threes{1}{1}{0} \threes{.}{.}{.} \threes{1}{1}{0}^\omega
\end{displaymath}
is not recognised, even though it formalises the correct addition of $4+4=8$. Instead, the following
holds: for every $x,y \in \Reals$ there is a $w \in \WF{\Sigma}{3} \cap R(\Aut{H}_{3=\mathsf{add}(1,2)})$ s.t.\ $\dec(w)_1 = x$, 
$\dec(w)_2 = y$ and $\dec(w_3) = x+y$. To capture the entire $R^3_{\mathsf{add}}$ we can use an equality automaton (here seen
as a transducer with 1 input and 1 output) to extend the output track to all representations of the same number: 
$\Aut{A}^3_{3=\mathsf{add}(1,2)} := \Aut{H}_{3=\mathsf{add}(1,2)} \circ \Aut{A}^2_{1=2}$. Note that $\circ$ here is relation composition,
not automaton intersection.

This covers the case of $k=2$ and summation of two numbers only. Any $\Aut{A}^k_{j=i_1{+}i_2}$ for other values of $i_1,i_2,j,k$ is 
easily obtained from $\Aut{A}^3_{3=\mathsf{add}(1,2)}$ by inserting and rearranging tracks, for instance using joins and projections. 
Automata formalising the addition of multiple summands can be obtained by breaking it down into sums of two values each:
\begin{displaymath}
	\Aut{A}^{k}_{j=\mathsf{add}(i_1,\ldots,i_n)} := \proj{\big(\Aut{A}^{k+1}_{k+1=\mathsf{add}(i_1,\ldots,i_{n-1})} \cap \Aut{A}^{k+1}_{j=\mathsf{add}(k{+}1,i_n)}\big)}{1,\ldots,k}
\end{displaymath}  
The size estimation results from the $k$-fold product constructions underlying the join operations on automata of constant size 
for constant numbers of tracks. \qed

\subsubsection*{Lemma \ref{lem:help}.\ref{lem:help;relu}.}
Note that the ReLU operation is the equality relation on positive inputs and the constant 0 on negative inputs. It is easy to
modify $\Aut{A}^k_{j=i}$ from Lemma~\ref{lem:equality} to obtain $\Aut{A}^k_{j=\mathsf{relu}(i)}$; here we briefly discuss how 
to do so in the case of $k=2$, $i=1$, $j=2$, as this can be matched directly to the picture in Fig.~\ref{fig:equality}: the 
transitions from the initial state $q_0$ under $[-,-]$ and $[-,+]$ are removed, and two new states are added to allow it
to also accept words $w \in \WF{\Sigma}{2}$ s.t.\ $w_1$ is of the form $-\{0,1\}^*.\{0,1\}^\omega$ and $w_2$ is of the form
$\{+,-\}0^*.0^\omega$. 
%
The cases of $k,i,j$ having different values are then easily obtained by extending the transitions with arbitrary
bits in the $k-2$ components other than $i$ and $j$. 
\qed

\subsubsection*{Lemma \ref{lem:help}.\ref{lem:help;const_mult}.}
First we consider integer values of $c$. The case of $c=1$ is just an instance of the equality automaton, and the case of $c=0$ is easy to construct
similarly. So suppose that $c=2$. Note that
$\Aut{A}^k_{j=\mathsf{mult}(2,i)} := \proj{(\Aut{A}^{k+1}_{k{+}1=i} \cap \Aut{A}^{k+1}_{j=\mathsf{add}(i,k+1)})}{(1,\ldots,k)}$ provides the desired functionality of
checking whether the $j$-th track contains double the value of the $i$-th track. Note that $\size{\Aut{A}^k_{j=\mathsf{mult}(2,i)}} = \BigO{1}$.

Now let $c \geq 2$ be an integer value. Let $m$ be minimal and $b_0,\ldots,b_{m-1}$ be chosen uniquely s.t.\ $c = \sum_{i=0}^{m-1} b_i \cdot 2^i$. 
Let $i_1,\ldots,i_\ell$ be the sequence of indices $i$ s.t.\ $b_i=1$. Then $\Aut{A}^k_{j=\mathsf{mult}(c,i)} :=$
\begin{displaymath}
	\proj{\Big(\Aut{A}^{k+m}_{k+1=\mathsf{mult}(2,i)} \cap \big(\bigcap\limits_{i=2}^{m-1} \Aut{A}^{k+m}_{k+i=\mathsf{mult}(2,k+i-1)}\big) \cap \Aut{A}^{k+m}_{j=\mathsf{add}(i_1,\ldots,i_\ell)}\Big)}{1,\ldots,k}
\end{displaymath}
recognises the relation stated in the lemma by essentially computing successive values $2^i \cdot x$ in the $i$-th
additional track (that gets projected out afterwards) for a value $x$ in the $i$-th input track, and then adding up these additional tracks 
that correspond to multiples of $x$ whose bit in the binary representation of $x$ is set. Note that $\size{\Aut{A}^k_{j=\mathsf{mult}(c,i)}} = 2^{\BigO{\log c}} = 2^{\BigO{\measure{c}}}$.

Multiplication with a negative integer constant $c$ can be realised by taking $\Aut{A}^k_{j=\mathsf{mult}(-c,i)}$ for the positive constant 
$-c$, and then changing the transitions out of its initial state by swapping the labels $+$ and $-$ in the $i$-th component. 

So suppose now that $c$ is genuinely rational. Take $n \in \Ints$, $1 \le d \in \Nats$ 
to be minimal such that $\frac{n}{d} = c$. Note that, for given $x,y \in \Reals$, we have $y = c \cdot x$ iff $n\cdot x = d \cdot y$.
Hence,
\begin{displaymath}
	\Aut{A}^k_{j=\mathsf{mult}(c,i)} := \proj{\big(\Aut{A}^{k+2}_{k+1=\mathsf{mult}(n,i)} \cap \Aut{A}^{k+2}_{k+2=\mathsf{mult}(d,j)} \cap \Aut{A}^{k+2}_{k+2=k+1}\big)}{1,\ldots,k}
\end{displaymath}
recognises multiplication with the constant $c$.

Remember that $\measure{c} = \BigO{\log |n| + \log d}$, and that the automata in this construction are of size exponential in 
the number of their tracks which is bounded by $\log |n|$, resp.\ $\log d$. Hence, $\size{\Aut{A}_{\mathsf{mult}}^c} = 2^{\BigO{\measure{c}}}$. 
\qed

\subsubsection*{Lemma \ref{lem:help}.\ref{lem:help;const}.}
The cases of $c=0$ and $c=1$ are easily constructed specifically. For instance, if $c=1$ then $\Aut{A}^k_{i=\mathsf{const}(1)}$ only has to check 
that the $i$-th track of an input contains $+0^*1.0^\omega$ or $+0^*.1^\omega$. So suppose that $0 \ne c \ne 1$. Then $c^{-1} \in \Rats$, 
$c^{-1} \ne 1$ and $\Aut{A}^k_{i=\mathsf{const}(c)}$ can be obtained as $\proj{(\Aut{A}^{k+1}_{k+1=\mathsf{mult}(c^{-1},i)} \cap \Aut{A}^{k+1}_{k+1=\mathsf{const}(1)})}{1,\ldots,k}$. Note that $\measure{c^{-1}} = \measure{c}$. \qed

\subsection{Use Cases: Analysing DNN using \FGWNBA}

\begin{figure}
	\centering
	\begin{tikzpicture}[semithick,initial text={\normalsize$\Aut{A}^2_{1<2}$}, every state/.style={inner sep=2pt, minimum size=4mm},font=\scriptsize,
		node distance=19mm]
		\node[state,initial,initial where=left]  (s0)                     {};
		\node[state]                             (s1) [above right of=s0]       {};
		\node[state]                             (s2) [below right of=s0] {};
		\node[state]                             (s3) [below right of=s1] {};
		\node[state]                             (s4) [above right of=s3] {};
		\node[state]                             (s5) [below right of=s3] {};
		\node[state,accepting]                   (s6) [below right of=s4]        {};
		
		\path[->] 
		(s0) edge              node[below left=-1mm]        {$\twos{-}{-}$} (s2)
		(s0) edge              node[above left=-1mm]        {$\twos{+}{+}$} (s1)
		(s1) edge              node[above]        {$\twos{0}{1}$} (s4)
		(s1) edge [loop above] node [left]                  {$\twos{0}{0}\!\!,\!\!\twos{1}{1}\!\!,\!\!\twos{.}{.}$} ()
		(s2) edge              node[above]        {$\twos{1}{0}$} (s5)
		(s2) edge [loop below] node [left]                  {$\twos{0}{0}\!\!\twos{1}{1}\!\!\twos{.}{.}$} ()
		(s0) edge              node[above]        {$\twos{-}{+}$} (s3)
		(s4) edge [bend left]  node[above right=-1mm]        {$\twos{0}{0}\!\!,\!\!\twos{0}{1}\!\!,\!\!\twos{1}{1}$} (s6)
		(s4) edge [loop above] node [right]                  {$\twos{1}{0}\!\!,\!\!\twos{.}{.}$} ()
		(s5) edge [bend right] node[below right=-1mm]        {$\twos{0}{0}\!\!,\!\!\twos{1}{0}\!\!,\!\!\twos{1}{1}$} (s6)
		(s5) edge [loop below] node [right]                  {$\twos{0}{1}\!\!,\!\!\twos{.}{.}$} ()
		(s3) edge              node[above]        {$\twos{0}{1}\!\!,\!\!\twos{1}{0}\!\!,\!\!\twos{1}{1}$} (s6)
		(s3) edge [loop above] node [left] {$\twos{0}{0}$} node [right] {$\twos{.}{.}$} ()
		(s6) edge [loop right] node [right]                  {$\twos{0}{0}\!\!,\!\!\twos{0}{1}\!\!,\!\!\twos{1}{0}\!\!,\!\!\twos{1}{1}\!\!,\!\!\twos{.}{.}$} ();
	\end{tikzpicture}
	\caption{Automaton recognizing that the number encoded on the first track is smaller than the number encoded
		on the second track.}
	\label{fig:smalleraut}
\end{figure}
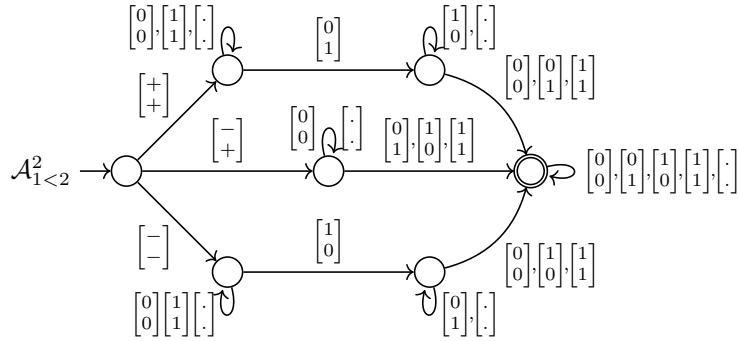
To specify automata recognizing the validity ARPs, we need to formalise comparison of vector components using \FGWNBA.
\begin{lemma}
	\label{sec:verification;lem:smaller}
	\label{lem:smaller}
	Let $k \geq 2$, $1 \le i,j \le k$. There are \FGWNBA $\Aut{A}^k_{i<j}$ and $\Aut{A}^k_{i\le j}$ of size $\BigO{1}$ s.t.\ for all 
	$w \in \WF{\Sigma}{k}$ we have $w \in R(\Aut{A}_<)$ iff $\dec(w_i) < \dec(w_j)$, resp.\ $\dec(w_i) \le \dec(w_j)$.    
\end{lemma}
\begin{proof}
	We show how $\Aut{A}^k_{i<j}$ can be built for $k=2$, $i=1$ and $j=2$ in Fig.\ref{fig:smalleraut}.
	It checks that the word on the second track encodes a greater number, depending on the sign, in the 
	obvious way. If the preceding signs are $[-,+]$ then it only needs to check that not both tracks encode $0$.
	If they are $[+,+]$ then the automaton needs to verify that the tracks differ at some point, and that, at the first
	point where they differ, the bit in the second track is set and the bit in the first track is not set. Moreover, the tracks
	can not continue with all following bits set in the first track, but none in the second, because then the numbers encoded in
	the tracks would be the same. Again, by padding the
	transition labels accordingly, one can create $\Aut{A}^k_{i<j}$ for arbitrary $k,i,j$.
	$\Aut{A}^k_{i\le j}$ is then simply obtained as $\Aut{A}^k_{i<j} \cup \Aut{A}^k_{i=j}$. All involved automata are of constant
size. \qed
\end{proof}
Note the slight difference in the specification in Lemma~\ref{lem:smaller} compared to the lemmas in Sect.~\ref{sec:translation}.
While the automata constructed there only accept well-formed words, the ones constructed in Lemma~\ref{lem:smaller} also accept
non-well-formed words. It would be easy to restrict the languages of $\Aut{A}^k_{i<j}$ and $\Aut{A}^k_{i\le j}$ to well-formed
words only by doubling the state space. This is, however, not necessary as they will only be used here in conjunction with other
\FGWNBA that ensure well-formedness.

\subsubsection*{Lemma \ref{lem:arp_in}.}
We start by arguing that one can construct a \FGWNBA $\Aut{A}^k_{j=\mathsf{abs}(i)}$ (of constant size) that 
checks whether the $j$-th track in a $k$-track word contains the absolute value of the number encoded in the $i$-th track.
It is easily obtained by swapping two transitions in $\Aut{A}^k_{j=i}$, namely those out of the initial state with labels
$[-,+]$ and $[-,-]$.

$\Aut{A}^{k,P}_{\mathsf{in}}$ can then be built by temporarily using $4m+2$ tracks in addition to the $k$ given ones which are 
checked to contain, respectively, for input values $x_1,\ldots,x_m$ encoded on the first $m$ tracks, the values $-x_1,\ldots,-x_m$, 
then the values $r_1,\ldots,r_m$, then $r_1 - x_1, \ldots, r_m - x_m$, then their absolute values in the next $m$ tracks, the sum 
of these in the next, and the constant $d$ in the last. Using the \FGWNBA from Lemmas~\ref{lem:help}.\ref{lem:help;add}, \ref{lem:help}.\ref{lem:help;const_mult}, 
\ref{lem:help}.\ref{lem:help;const}, 
\ref{lem:smaller}, the correctness of the tracks can be verified as follows. Let $\ell := k+4m+2$. Then $\Aut{A}^{k,P}_{\mathsf{in}}$ 
is defined via
\begin{align*}
	\Aut{A}^{k,P}_{\mathsf{in}} := 
	\Big(&\big(\bigcap\limits_{i=1}^m \Aut{A}^\ell_{k+i=\mathsf{mult}(-1,i)} \cap \Aut{A}^\ell_{k+m+i=\mathsf{const}(r_i)} \cap  
	\Aut{A}^\ell_{k+2m+i=\mathsf{add}(k+i,k+m+i)}  \\ 
	&  \hphantom{\big(\bigcap\limits_{i=1}^m } \cap\Aut{A}^\ell_{k+3m+i=\mathsf{abs}(k+2m+i)}\big) \\
	& \cap \proj{\Aut{A}^\ell_{k+4m+1=\mathsf{add}(k+3m+1,\ldots,k+4m)} \cap \Aut{A}^\ell_{\ell=\mathsf{const}(d)} \cap \Aut{A}^\ell_{k+4m+1 \le \ell}\Big)}{1,\ldots,k}
\end{align*} 
This construction of $\Aut{A}^{k,P}_\text{in}$ makes the size claim obvious: the important parts are the addition, multiplication and
constant automata are exponential in their respective parameters, each a subparamter of $P$. The intersection of all these,
leads to the size of $2^{\BigO{\measure{P}}}$.
\qed

\end{document}